\documentclass{article}
\usepackage[utf8]{inputenc}
\usepackage{mathtools}
\usepackage[margin=1in,a4paper,pdftex]{geometry}
\usepackage{parskip}
\usepackage[backend=bibtex,style=alphabetic,url=false,doi=false,eprint=false,url=false,isbn=false]{biblatex}
\usepackage{csquotes}
\AtEveryBibitem{\clearfield{month}}
\AtEveryBibitem{\clearfield{day}}
\AtEveryBibitem{\clearfield{note}}
\usepackage{faktor}
\usepackage[english]{babel}
\usepackage[final]{pdfpages}
\usepackage{braket}
\usepackage{tensor}
\usepackage{xcolor}

\DeclareFontFamily{U}{dmjhira}{}
\DeclareFontShape{U}{dmjhira}{m}{n}{ <-> dmjhira }{}

\usepackage{amsmath, amssymb, amsthm,amsfonts}
\usepackage{slashed}
\usepackage{graphicx}
\usepackage{epstopdf}
\usepackage{enumerate}
\usepackage{color}
\usepackage{tikz}
\usepackage{mathrsfs}
\usepackage{tikz-cd}
\usepackage{hyperref}
\usepackage{pdfpages}
\usepackage{cleveref}

\newcommand{\T}{\mathbb{T}}

\newcommand{\Z}{\mathbb{Z}}

\newcommand{\R}{\mathbb{R}}

\DeclareMathOperator{\Hom}{Hom}
\DeclareMathOperator{\Tor}{Tor}
\DeclareMathOperator{\Ext}{Ext}

\theoremstyle{definition}
\newtheorem{theorem}{Theorem}[section]

\newtheorem{corollary}[theorem]{Corollary}
\newtheorem{proposition}[theorem]{Proposition}

\theoremstyle{remark}

\newtheorem{remark}[theorem]{Remark}

\title{Cohomology of knotted semimetals in three dimensions}
\author{Joshua Celeste}
\date{December 2023}

\begin{document}
	
\maketitle

\begin{abstract}
	We extend the topological classification scheme of Weyl semimetals via cohomology and the Mayer-Vietoris sequence to account for nodal line semimetals with space-time inversion symmetry. These are semimetals where bands meet generally in 1-dimensional submanifolds, which can generally be knots in $\T^3$. These nodal loops have two charges, the quantized Berry phase and the $\Z_2$-monopole charge, the second related to linking numbers of nodal knots between bands. We provide a manifestly topological proof of the Weyl charge cancellation condition for the $\Z_2$ monopole charge, which is known to be the second Stiefel-Whitney class of a tubular neighbourhood surrounding a Weyl submanifold via the Mayer-Vietoris sequence.
\end{abstract}

\section*{Acknowledgements}
This work is part of ongoing work towards a PhD thesis, as such I would like to acknowledge the financial support of the University of Adelaide, and thank my principal supervisor Mathai Varghese for helpful discussions throughout the production of this work.

\section{Introduction}
Weyl semimetals are systems that realise Weyl fermions in condensed matter systems, these were predicted in \cite{wanTopologicalSemimetalFermiarc2011a} and have gained much interest since their experimental discovery, such as in \cite{xuDiscoveryWeylFermion2015b}. Much of the initial focus has been on semimetals with band crossings at Weyl points, these systems exhibit interesting experimental signatures when a boundary is added, such as the presence of Fermi arcs joining projections of Weyl points on the boundary. However, there is an increasing focus on systems with band crossings at higher dimensional submanifolds, such as nodal semimetals, these systems also exhibit interesting surface states on the boundary, namely drumhead states and hinge states, and are expected to have interesting transport properties \cite{burkovTopologicalNodalSemimetals2011,matsuuraProtectedBoundaryStates2013,fangTopologicalNodalLine2015a,zhaoSymmetricRealDirac2017}.

Here we extend the previous work of \cite{mathaiDifferentialTopologySemimetals2017,mathaiGlobalTopologyWeyl2017a} to further account for certain classes of nodal semimetals, in particular nodal semimetals with space-time inversion symmetry $I_{ST}$ present. These were first described in \cite{fangTopologicalNodalLine2015a}, these systems possess nodal lines with two charges, the quantized Berry phase around each nodal loop, and a second charge known in the physics literature as the $\Z_2$-monopole charge. However, it was shown in \cite{ahnStiefelWhitneyClasses2019} that these simply correspond to the first and second Stiefel-Whitney numbers of the valence bundle for certain submanifolds related to the nodal lines. In this work, the topological properties of such nodal semimetals are analysed via the Mayer-Vietoris sequence. Consequently, we obtain a topologically manifest proof of the charge cancellation condition of the $\Z_2$ monopole charge for nodal semimetals via the Mayer-Vietoris sequence.

\section{Topology and semimetal phases}
\subsection{Topological invariants of insulating phases}
For any system amenable to study via the Bloch-Floquet transform, with some sort of a band gap present, the usual approach is to view the system as a family of Hamiltonians acting on a family of finite dimensional (usually complex) vector spaces that live over each point in the Brillouin zone.  We denote the Brillouin zone generally by $T$, although it is often an $n$-dimensional torus, in which case we shall simply write it as $\T^n$. When there is a band gap present everywhere at the Fermi level, we say the phase is \textit{insulating}, otherwise it is \textit{semimetallic}. We are interested in the behaviour of the system only up to the Fermi energy, the projection of the Hamiltonian up to the Fermi level, then defines a family of projections of which act on the aforementioned family of vector spaces, giving rise to another, possibly topologically non-trivial family of vector spaces. This is precisely the notion of a vector bundle, and a known tool to assist in distinguishing such objects are \textit{characteristic classes}. Such objects are elements of \textit{cohomology groups} associated to the Brillouin zone and provide topological obstructions to continuously deforming one family of Hamiltonians to another.

Many invariants of topological insulators and semimetals are cohomological in nature, and are typically derived from some expression of the Berry curvature. For instance, 2D Chern insulators are classified by the first Chern number which is simply the integral of the Berry curvature over the Brillouin zone. A more topological picture occurs when we view this instead as the pairing between the first Chern class $c_1(V)$ of the complex line bundle determined by the valence band, and the fundamental class of the Brillouin zone $[T]$. In this view, the Chern class lives in the \textit{cohomology group} $H^2(T;\Z)$ while the fundamental class $[T]$ is an element of the \textit{homology group} $H_2(T;\Z)$. 

Generally, there is a pairing between cohomology and homology groups $\braket{\cdot,\cdot}:H^n(T;\Z)\times H_n(T;\Z)\rightarrow \Z$, in particular for cohomology classes that are representable in \textit{deRham cohomology}, such as the first Chern class, this pairing is simply given by integrating an appropriate representative differential form (for instance the Berry curvature) over an oriented submanifold that represents the desired homology class in $H^n(T;\Z)$. Hence for a 2D Chern insulator, the Brillouin zone itself has a fundamental class that lives in $H^2(T;\Z)$, and its pairing with the first Chern class of the valence bundle gives rise to the first Chern number.

Since the valence line bundle is a complex line bundle bundle and such objects are known to be classified up to isomorphism by their first Chern class, it follows that the first Chern number is the only such cohomological invariant for the Chern insulator in 2 dimensions. However, in more complicated systems, higher dimensional invariants can be present, and if additional symmetries are allowed, the relevant cohomology groups may not be the ordinary integral cohomology groups, and \textit{equivariant cohomology} is required instead \cite{freedTwistedEquivariantMatter2013b,thiangKTheoreticClassificationTopological2016a}.

It is of note that often one does not simply wish to classify vector bundles up to isomorphism, but instead by the weaker notion of \textit{stable isomorphism}. In such a case there are other classification schemes, such as those involving K-theory, which gives rise to the known classification of topological insulators \cite{kitaevPeriodicTableTopological2009a}.

Not every invariant of a band system lives in cohomology or K-theory however. For instance, 2D semimetals, such as graphene have the Berry phase around a closed loop in the Brillouin zone being a topological invariant \cite{vanderbiltBerryPhasesElectronic2018}. However, this has no cohomological interpretation, instead it is better viewed as the holonomy of the Berry connection, more refined invariants such as these are elements of more exotic cohomology theories such as \textit{differential cohomology} \cite{satiAnyonicTopologicalOrder2023a} which we shall not describe here.

\subsection{Invariants of phases with a band intersection}
One can consider more general systems where a band gap is not present everywhere over the Brillouin zone and hence two bands intersect, usually at the Fermi energy. We shall refer to these points as a \textit{band intersection} or \textit{band crossing}, which determines some submanifold $W$ of the Brillouin zone $T$ called the \textit{Weyl submanifold} of the system. Due to the presence of a band intersection at the Fermi level, the projection operators obtained from the family of Hamiltonians are not defined everywhere, only away from $W$, thus we only obtain a vector bundle over $T\setminus W$. As a consequence, for a cohomological study of these phases of matter, the appropriate groups to study are the cohomology groups $H^n(T\setminus W)$. Even when the Brillouin zone is the torus $T=\T^n$, determining such cohomology groups depends on $W$ and can be quite non-trivial.

\subsubsection{Topology of Weyl Semimetals}
In this section we briefly discuss the topological classification of Weyl semimetals as seen in \cite{mathaiDifferentialTopologySemimetals2017,mathaiGlobalTopologyWeyl2017a} These are semimetallic phases where band crossings occur at points in momentum space. The case of most relevant interest is a simple two-band model in three dimensions. Consider a family of Hamiltonians of the form
$$H(k)=\mathbf{h}(k)\cdot \sigma$$
where $\mathbf{h}(k)$ is some smooth function $\mathbf{h}:\T^{3}\rightarrow\R^3$ and $\sigma=(\sigma_x,\sigma_y,\sigma_z)$ is a vector containing the three Pauli matrices. The spectrum is given by $\pm|\mathbf{h}|$ and so if one sets the Fermi level to be zero, band crossings occur at the zeroes of $\mathbf{h}$, which are generically points, such points are called \textit{Weyl points}. We assume for this section therefore that $\mathbf{h}$ has finitely many isolated zeroes $k_i$ with the Weyl submanifold being their union $W=\bigcup_{i}k_i$.

Although a band gap does not occur everywhere, the system is gapped along submanifolds of the Brillouin zone which avoid $W$ and thus the Berry curvature is well defined away from $W$. For each Weyl point, one defines a \textit{local charge} by integrating the Berry curvature on a sphere around a point. Alternatively, one can normalise the map $\mathbf{h}$ to give a map $\widetilde{\mathbf{h}}:\T^3\rightarrow S^2$. Restricting $\widetilde{\mathbf{h}}$ to a small sphere around each Weyl point $k_i$ gives a collection of maps $\widetilde{\mathbf{h}}_i:S^2\rightarrow S^2$, each map naturally having a degree which can be viewed as the local charge of the $i$th Weyl point. However, it can be shown this is related to a more cohomological picture. Namely, that the local charge of $k_i$ can be viewed as the pullback under $\widetilde{\mathbf{h}}_i$ of the first Chern class of the hyperplane line bundle over the surrounding sphere $S^2$.

These local charges provide obstructions to being able to continuously deform the Hamiltonian so that it becomes a gapped system, and thus a topological insulator. Furthermore, we can consider the restriction of the valence bundle to subtori of $\T^3$, typically coordinate subtori defined by fixing the value of one of the coordinates $k_x,k_y$ or $k_z$ of the Brillouin zone. These have three independent first Chern classes  $c_1(V)_x$, $c_1(V)_y$ and $c_1(V)_z$ , or equivalently, Chern numbers $c_{1x}, c_{1y},c_{1z}$ by pairing the cohomology classes with the coordinate subtori, which remain constant due to continuity, until a Weyl point is passed. Due to Stokes' theorem, there is a jump in the Chern numbers by the value of the local charge passed. It is also this behaviour which gives rise \textit{Fermi arcs} when a boundary is introduced that join the Weyl points. The topological nature of these were also studied in \cite{mathaiDifferentialTopologySemimetals2017,mathaiGlobalTopologyWeyl2017a} via Euler chains.

 However, as their name suggests, local charges are only local in nature, and do not provide the full picture. For the local charges to be Chern classes of the valence bundle restricted to their surrounding spheres, a global consistency condition is required which can be derived from the Mayer-Vietoris sequence in cohomology \cite{bottDifferentialFormsAlgebraic2008}. The Mayer-Vietoris sequence relates the spaces $\T^3\setminus W$, a tubular neighbourhood of $W$, which we shall denote by $D_W$ and the their intersection, which is a space homotopic to the spheres surrounding the Weyl points described above, denoted $S_W$. We obtain an exact sequence of maps, and of particular interest is the end of the sequence, namely
$$\dots\rightarrow H^{2}(\T^3)\xrightarrow{\iota_{\T^3\setminus W}^*}H^2(\T^3\setminus W)\xrightarrow{\iota_{S_W}^*} H^2(S_W)\xrightarrow{\Sigma} H^3(\T^3)\rightarrow 0$$
where $\iota_{\T^3\setminus W}$ and $\iota_{S_W}$ are the maps on cohomology induced by inclusion. Meanwhile $H^2(S_W)$ decomposes as a direct sum $H^2(S_W)\cong \bigoplus_{i}H^2(S_{(i)})\cong \bigoplus_{i}\Z_{(i)}$ and $H^3(\T^3)\cong\Z$, under these isomorphisms the map $\Sigma$ is simply the sum of the integers in $H^2(S_W)$, thus we view it as a 'sum of charges' map. 

Exactness means that the kernel of one map is the image of the proceeding map, since complex line bundles are classified by the integral second cohomology groups of the base space, it follows that a complex line bundle over $S_W$ (which corresponds to a collection of local charges) is in the image of $\iota_{S_W}^*$ i.e. the restriction of of a complex line bundle over $\T^3\setminus W$ and hence is the line bundle of a Weyl semimetal, if and only if the local charges are in the kernel of $\Sigma$. Since $\Sigma$ can be viewed as the sum of charges, this implies that local charges in a Weyl semimetal must satisfy a charge cancellation condition, that is, the local charges must all sum to zero.
 
 \begin{figure}[h]
 \begin{center}
 	\includegraphics[scale=0.25]{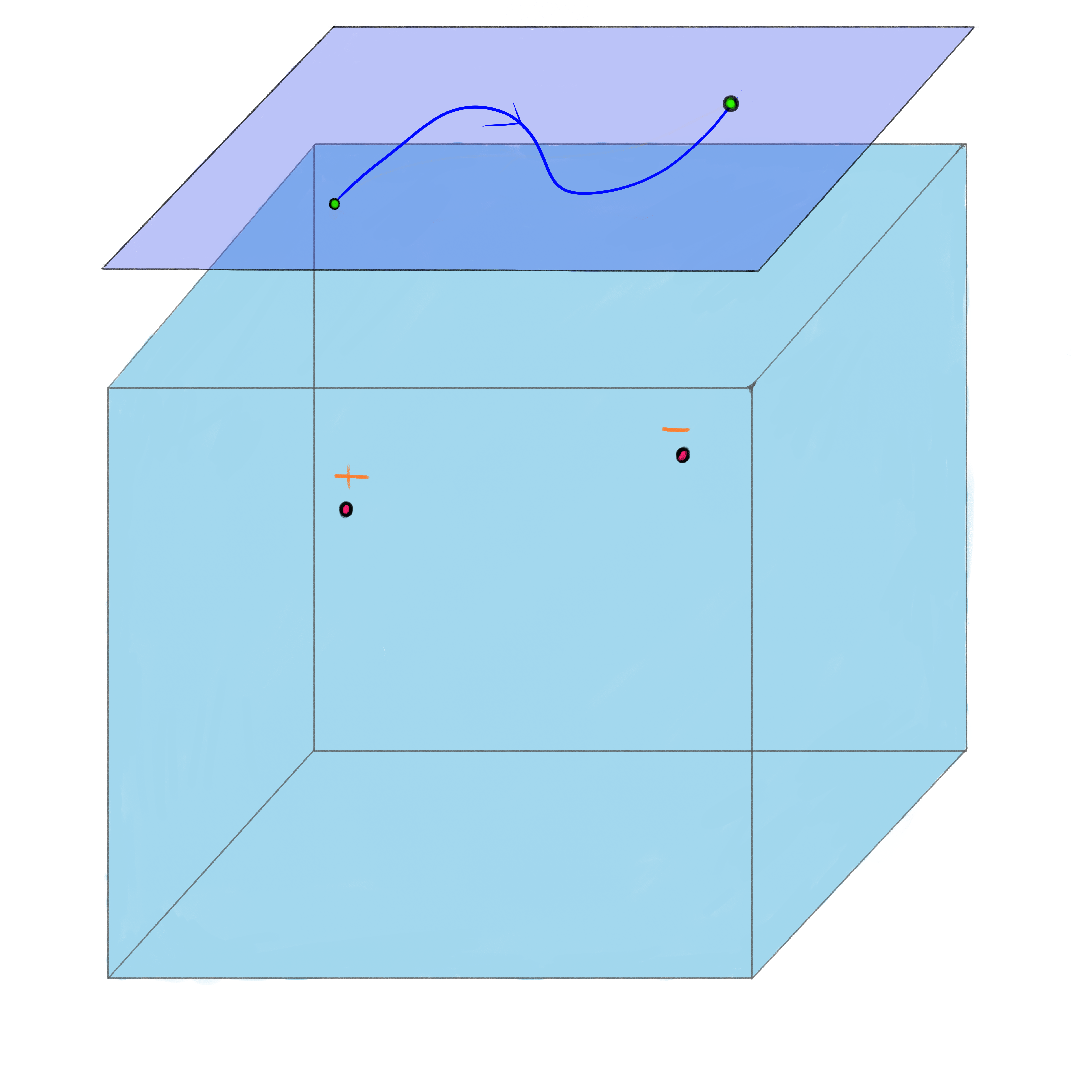}
 	\caption{Representation of a Weyl semimetal with two Weyl points, one of charge $+1$ and the other with $-1$. This is required by the charge cancellation condition. Above is a depiction of a Fermi Arc which occurs when a boundary is a continuous curve of states on the boundary joining the two Weyl points.}
 \end{center}
 \end{figure}
\newpage
 
 \subsection{The Mayer-Vietoris sequence in the general case}
 Here we investigate the Mayer-Vietoris sequence for Weyl submanifolds of general dimension in $\T^n$. A general two band model in dimension $n$ has the form
 $$H(k)=\sum_{i=1}^{N}h^i(k)\sigma_i$$ where $\sigma_n$ are representations of the generators of the complex Clifford Algebra. Such a system is essentially specified by the real-valued function $\mathbf{h}$ on $\T^n$ with components $h^i(k)$.
 
 One could consider more complicated systems, such as one with multiple bands, nonetheless, the following argument is purely topological and still applies. Assume that $\mathbf{h}$ has zeroes on some submanifold $W$ which decomposes into disjoint connected submanifolds $W_i$. Furthermore, we may surround each $W_i$ with a tubular neighbourhood $D_i$ with boundary $S_i$, importantly, these boundaries $S_i$ are of dimension $n-1$. Denote the disjoint union of these submanifolds by $W$, $D_W$ and $S_W$ respectively, then set $U=\T^n\setminus W$ and $V=D_W$. Applying the Mayer-Vietoris sequence then gives the following result
  
  \begin{proposition}
  	Let $W=\sqcup_{i}W_i$ and $D_i$ be as above, then the following sequence is exact
  	$$\dots\rightarrow H^p(\T^n)\rightarrow H^p(\T^n\setminus W)\oplus \bigoplus_{i}H^p(W_i)\rightarrow \bigoplus_{i}H^p(S_i)\xrightarrow{\Sigma}H^{p+1}(\T^n)\rightarrow0$$
  	where $\Sigma$ is the boundary map of the Mayer-Vietoris sequence.
  \end{proposition}

in particular  for $n>2$ we have
$$\dots \rightarrow H^{n-1}(\T^n)\xrightarrow{\iota} H^{n-1}(\T^n\setminus W)\xrightarrow{\iota} \bigoplus_{i}H^{n-1}(D_i)\xrightarrow{\Sigma}H^n(\T^n)\rightarrow 0 .$$ Since the $D_i$ are $n-1$ dimensional compact connected submanifolds of $\T^n$, it follows that $H^{n-1}(D_i)\cong \Z_{(i)}$ which we label with $i$ to keep track of the index $i$, we also have $H^n(\T^n)\cong \Z$. Under these isomorphisms, the map $\Sigma$ is equivalent to a map $\Sigma:\bigoplus_{i}\Z_{(i)}\rightarrow\Z$ defined by
$$\Sigma(x_1,\dots x_k)=\sum_{i=1}^kx_i.$$

To prove this it is easiest to see with deRham cohomology, which we may use since the last two remaining terms in the integral cohomology groups are torsion-free.

\begin{proof}
	Suppose that $\omega$ is an $n-1$ form representing a class in $H^{n-1}(U\cap V)$. Under the isomorphism $H^{n-1}(U\cap V)\cong H^{n-1}(S_W)$ it suffices to consider an $n-1$ form defined on $S_W$ which is the restriction of $\omega$ to the boundary $S_W=U\cap V$. Since $S_W$ decomposes as a disjoint union, the cohomology group can be written as a direct sum
	$$H^{n-1}(S_W)\cong \bigoplus_{i}H^{n-1}(S_i)$$
	so we may write $\omega=\sum_{i}\omega_i$ where each $\omega_i$ has support in $S_i$, the corresponding cohomology class is $([\omega_1],\dots[\omega_{\ell}])$. We shall further take the restriction of $\omega_i$ to the boundary $H^{n-1}(S_i)$ and assume its integral is 1 so that it is a generator for the $i$-th component of $H^{n-1}(S_W)$. Let $\rho_{U},\rho_{V}$ be a partition of unity subordinate to ${U,V}$. The boundary map of the Mayer-Vietoris sequence in deRham cohomology is defined by
	\begin{align*}
		\Sigma(\omega)=\begin{cases}
			& [-d(\rho_V\omega)] \text{ on U} \\
			& [d(\rho_U\omega)]\text{ on V}
		\end{cases}
	\end{align*}
this is a globally defined form on $\T^n$ supported in $U\cap V$, since $H_{\text{dR}}^n(\T^n)\cong\R$ with the isomorphism given by integration over $\T^n$, we compute
\begin{align*}
	\int_{\T^n}\Sigma(\omega)&=\int_{U\cap V}d(\rho_U\omega) \\
	&=\int_{\partial(U\cap V)}\rho_U\omega \\
	&=\int_{S_W}\rho_U\omega.
\end{align*}
However, $S_W$ is contained entirely in $V$ where $\rho_U=1$, hence
\begin{align*}
	\int_{S_W}\rho_U\omega&=\int_{S_W}\omega \\
	&=\sum_i\int_{S_i}\omega_i
\end{align*}
and so under the identification $[\omega]=([\omega_1],\dots,[\omega_{\ell}])$, we see that
$$\Sigma(\omega_1,\dots,\omega_{\ell})=\sum_{i}\omega_i.$$
\end{proof}
 \begin{remark}
 	Note that this result technically holds in deRham, not integral cohomology. Provided all the singular cohomology groups in the Mayer-Vietoris sequence are free, there is an injection $H^q(X)\hookrightarrow H^q(X)\otimes R\cong H_{dR}^q(X)$.
 	Consequently, we may choose the generators in the above proof to have integral periods, and thus represent elements in singular cohomology.
 \end{remark}
 \begin{corollary}
 	As a consequence of exactness of the above sequence, cohomological local charges in degree $n-1$ (elements of $H^{n-1}(D_i)$) arise from degree $n-1$ characteristic classes of vector bundles over $\\T^n\setminus W$ (elements of $H^{n-1}(\T^n\setminus W)$), if and only if the local cohomological charges are in the kernel of the boundary map. That is, the charges sum to zero.
 \end{corollary}
 
 If nodal lines had their own cohomological charge in degree $d-1$, then it would follow from the above proposition that they must satisfy a charge cancellation condition identical to that of Weyl semimetals. However, in the key application of interest, namely $d=3$, this is not the case. To obtain nodal lines in $3D$ generically requires one of the components of $\mathbf{h}$ to be zero, and so the corresponding map into the 2-sphere after normalisation cannot be surjective, and thus has degree zero. So the standard procedure to obtain a cohomological invariant does not work for nodal lines. On the other hand, the standard local invariant for nodal lines is the berry phase around a loop containing the nodal line, which is only well-defined modulo $2\pi$, and thus cannot be interpreted as an element of integral cohomology.
 \begin{figure}[h]
 	\begin{center}
 		\includegraphics[scale=0.25]{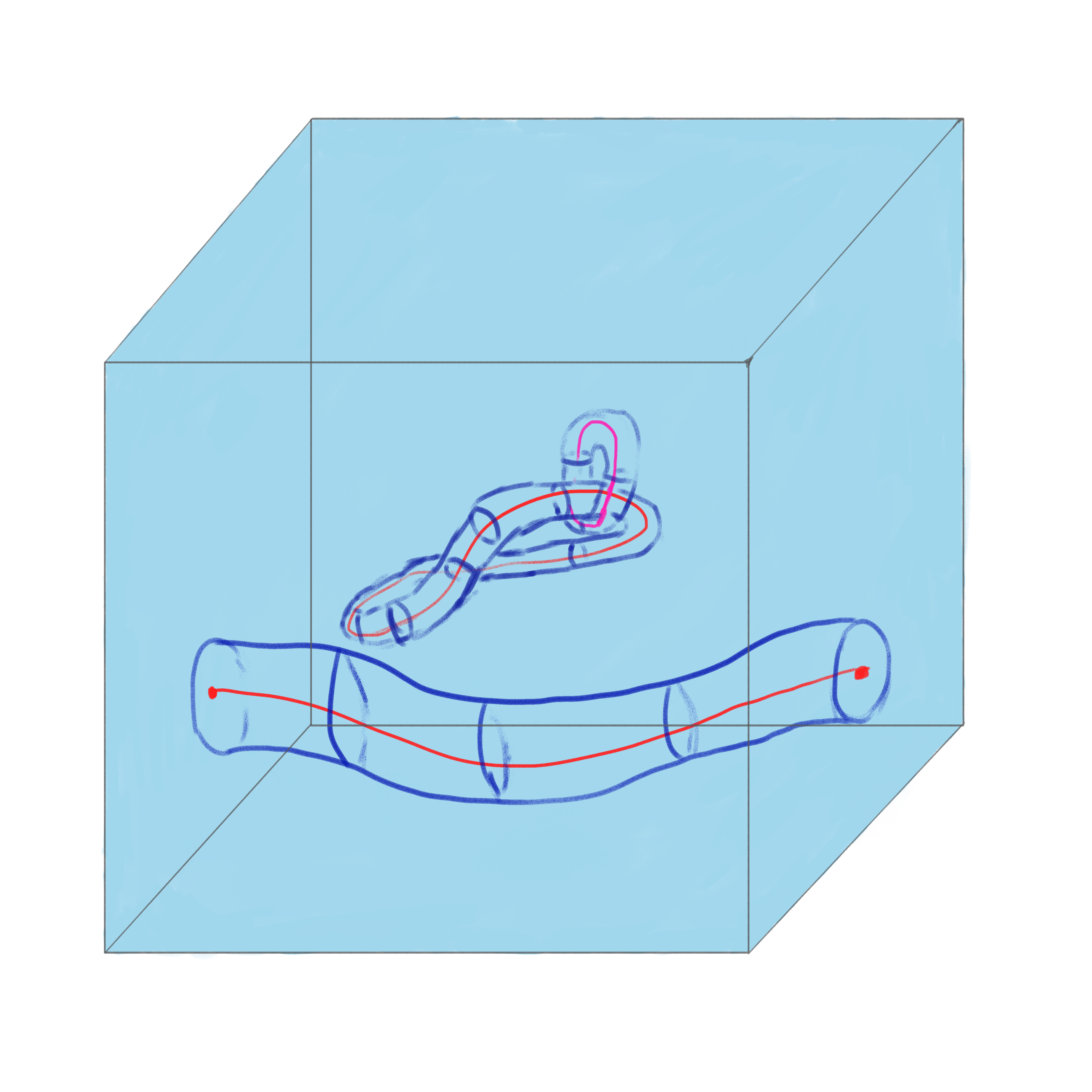}
 		\caption{Depiction of tubular neighbourhoods of a nodal link in the Brillouin torus.}
 	\end{center}
 \end{figure}
 \newpage
\section{Semimetallic phases with a real symmetry}
\subsection{Examples of Hamiltonians with a real symmetry}
One means to force the Weyl submanifold to have dimension 1 when the Brillouin zone is 3 dimensional, is to force the Hamiltonian to be real. As long as spin-orbit coupling is not considered, this can be achieved with space-time inversion symmetry $I_{ST}$. Further details on this symmetry and its connection to the Stiefel-Whitney classes is further discussed in \cite{avronTopologicalInvariantsFermi1988}. This an anti-unitary symmetry satisfying $I_{ST}^2=1$, and can be obtained in a 3D semimetal as $I_{ST}=P\circ T$ where $P$ spatial inversion and $T$ is time-reversal. Hence in position space this is simply the map $(t,\vec{r})\mapsto (-t,-\vec{r})$.

In the presence of space-time inversion symmetry, one can choose a basis in which it is represented as complex conjugation\cite{fangTopologicalNodalLine2015a,kimDiracLineNodes2015}. In such a case, the Hamiltonian and the states of the system are forced to be real. Consequently, one is left to classify \textit{real} vector bundles over $\T^n\setminus W$ as opposed to complex vector bundles. Moreover, since a two-band semimetal Hamiltonian generally has the form
$$H(k)=\mathbf{h}(k)\cdot\vec{\sigma}$$
in the presence of such a symmetry, coefficients of $\mathbf{h}$ which correspond to $\sigma_i$ which do not have purely real entries are forced to be zero. In the case of 3 dimensions, this forces $h^2=0$. Giving rise to the Hamiltonian
$$H(k)=h^1(k)\sigma_1+h^3(k)\sigma_3.$$
which generically has a 1 dimensional zero set, which in the most general case can be a knot or a link. As such we refer to these as \textit{knotted semimetals}, although they are most often referred to as nodal loop or nodal line semimetals in the physics literature.

Since the corresponding vector bundle over $\T^3\setminus W$ is real, the appropriate characteristic classes for their classification are the Stiefel-Whitney classes. These live in cohomology with $\Z_2$ coefficients and thus give rise to numbers which are $\Z_2$ valued when paired with homology classes. Just as with the pairing of Chern classes with distinguished homology classes to give rise to the local charges of Weyl points, it was shown in\cite{ahnStiefelWhitneyClasses2019} that the first and second Stiefel-Whitney classes give rise to two charges of nodal loops in momentum space, namely the quantized Berry phase and the $\Z_2$ monopole charge respectively. The first Stiefel-Whitney class can be paired with a loop that encloses a nodal line, and measures the non-orientability of the valence band restricted to this loop.

One should keep in mind that generically in a two-band system, the only invariant is the first Stiefel-Whitey class. To obtain a nodal semimetal with non-zero $\Z_2$-monopole charges, one must consider systems with more than two bands, for instance one described by a Hamiltonian of the form
$$H(k)=k_x\mathbf{1}\otimes\sigma_x+k_y\tau_y\otimes\sigma_y+k_z\mathbf{1}\otimes\sigma_z+m\tau_z\otimes\sigma_z$$
where the $\sigma_i,\tau_i$ are Pauli matrices. The spectrum is given by $E=\pm\sqrt{k_x^2+(\rho\pm|m|)^2}$ with $\rho=\sqrt{k_y^2+k_z^2}$. This yields a system with two linked nodal loops at $k_x=0$ and $\rho=0$ obtained from the crossing of two distinct sets of occupied bands.


The $\Z_2$ monopole charge is an obstruction to deforming the nodal loop to a point, which could then be further made to vanish, while mathematically, the second Stiefel-Whitney class of a vector bundle is an obstruction to defining a spin structure on said vector bundle, and is a cohomology class on the base space on the bundle.

\subsection{A proof of charge cancellation}

It is claimed that the $\Z_2$ monopole charge in particular has a mod 2 charge cancellation condition. However, in this section we provide a manifestly topological proof of this fact in dimension 3 by extending the Mayer-Vietoris sequence in integer cohomology to the $\Z_2$ case.

The idea is generally simple, namely to apply the well-known universal coefficient theorem to the integral cohomology sequence. However, there is subtelty required in justifying this application and using it to ensure an exact sequence of the appropriate cohomology groups, and that the resulting final boundary map is indeed the sum of charges modulo 2.

The following is a general statement of the universal coefficient theorem for free chain complexes \cite[Ch.\ 5,\ Thm.\ 10]{spanierAlgebraicTopology1995}
\begin{theorem}
	Let $C$ be a free chain complex such that $H(C)$ is of finite type or $G$ is finitely generated, then there is a functorial short exact sequence
	$$0\rightarrow H^q(C;R)\otimes G\xrightarrow{\mu}H^q(C;G)\rightarrow \Tor(H^{q+1}(C;R), G)\rightarrow 0 $$
	and this sequence also splits.
\end{theorem}
Note that the map is induced by the map $\mu:\Hom(A,G)\otimes G'\rightarrow \Hom (A,G\otimes G')$ acting on singular cochain complexes and is given by
$$\mu(f\otimes g')(a):=f(a)\otimes g'$$
where $f\in \Hom(A,G),g\in G'$ and $a\in A$.

We apply this proposition in the case when $R=\Z$ and $G=\Z_2$. Under the condition that all the cohomology groups are free, the torsion term vanishes in the exact sequence. It follows that $H^q(C;\Z)\otimes \Z_2\cong H^q(C;\Z_2)$. Furthermore, taking the tensor product is functorial. It follows that tensoring a sequence of $\Z$-modules by $\Z_2$ produces another sequence of $\Z$-modules.with the maps induced by taking the tensor product with the identity on the $\Z_2$ factor.

It is of note, that the tensor product is not always an exact functor. The stereotypical example of this behaviour being the fact that the sequence
$$0\rightarrow Z\xrightarrow {\cdot 2x}\Z\rightarrow \Z_2\rightarrow0$$ is exact, but the tensoring this sequence by $\Z_2$ only gives that the sequence
$$ \Z_2\xrightarrow {0} \Z_2\xrightarrow{\text{id}} \Z_2\rightarrow 0$$
with the leftmost map no longer injective.

In order to justify exactness we shall use a fairly standard of result of commutative algebra, consider Proposition 1.2.6 of \cite{liuAlgebraicGeometryArithmetic2010a} from which a fairly elementary fact of commutative algebra follows, namely the following proposition
\begin{proposition}
	Let $0\rightarrow M_1\rightarrow M_2\rightarrow\dots\rightarrow M_k\rightarrow 0$ be a short exact sequence of free $\Z$-modules and $N$ an arbitrary $\Z$-module. Then the induced sequence given by tensoring with $N$
	$$0\rightarrow M_1\otimes N\rightarrow\dots \rightarrow M_k\otimes N\rightarrow 0$$
	is also exact.
\end{proposition}

\begin{proof} 	Since the sequence is finite, one may break up the original long exact sequence into a series of short exact sequences, starting from the right. Since the right-most two modules are free, the left most entry is also free. Proceeding by induction yields the result.
\end{proof}

\begin{remark}
	This result holds more generally over for modules over any ring if one replaces 'free' with 'flat', although since we are only considering $\Z$-modules, the two notions are equivalent and we choose to work with the former for simplicity.
\end{remark}

It remains to show that each term in the Mayer-Vietoris sequence is free. Assume that $X$ is any compact, oriented topological space. It is known that $H^0(X)$ is simply a free $\Z$-module with rank counting the number of connected components of $X$ and is thus free \cite{hatcherAlgebraicTopology2002}.  The universal coefficient theorem for finitely generated abelian groups \cite{spanierAlgebraicTopology1995} gives a non-canonical isomorphism
$$H^p(X)\cong \Hom(H_p(X),\Z)\oplus \Ext(H_{p-1}(X),\Z)$$
from which it follows that $H^1(X)$ is free, since the Ext term is zero and the Hom term is always free, Poincar\'e duality gives $H^2(X)\cong H_1(X)$. 

	
Now $H_1(X)$ is not generally torsion free. However, for the spaces we are interested in this is the case. This fact essentially follows since $H_1$ is the abelianisation of the fundamental group. In the case of $\T^3\setminus W$, one picks up extra $\Z$ factors from each nodal loop in the fundamental group, still producing a free first homology group, and $U\cap V$ for nodal loops is homotopic to a disjoint union of 2-dimensional tori, each having $\Z^2$ as their first homology, and the union being their direct sum.

Finally, the only group of degree $3$ in the Mayer-Veitoris sequence in 3 dimensions is $H^3(\T^3)$, which is isomorphic to $H_0(\T^3)$ by Poincar\'e duality, hence torsion free also.

It follows from the above argument that the following sequence, obtained by taking the tensor product of the Mayer-Vietoris sequence in integral cohomology with $\Z_2$
$$\dots\xrightarrow{\partial} H^2(\T^3;\Z_2)\rightarrow H^2(\T^3\setminus W;\Z_2)\rightarrow\bigoplus_{i}H^2(\T^2_{(i)};\Z_2)\xrightarrow{\Sigma\text{ (mod 2)}}H^3(\T^3;\Z)\rightarrow 0$$
is exact.

 Furthermore, under the isomorphisms $\Z^k\otimes \Z_2\cong(\Z\otimes \Z_2)^k\cong \Z_2^k$ given by $(x_1,\dots,x_k)\otimes a\mapsto (x_1\otimes a,\dots ,x_k\otimes a)\mapsto (x_1 (\text{mod} 2),\dots x_k (\text{mod} 2))$.
 
 \begin{corollary}
 	Due to exactness elements of $\bigoplus_{i}H^2(\T^2_{(i)};\Z_2)$ come from the restriction of an element of $H^2(\T^3\setminus W;\Z)$, i.e. non-zero $\Z_2$ monopole charges arise from a vector bundle over $\T^3\setminus W$, if and only if they are in the kernel of the boundary map in the Mayer-Vietoris sequence. That is, if and only if the charges sum to zero modulo 2.
 \end{corollary}
 
 \newpage
 \printbibliography

@article{ahnStiefelWhitneyClasses2019,
  title = {Stiefel–{{Whitney}} Classes and Topological Phases in Band Theory},
  author = {Ahn, Junyeong and Park, Sungjoon and Kim, Dongwook and Kim, Youngkuk and Yang, Bohm-Jung},
  date = {2019-11},
  journaltitle = {Chinese Physics B},
  shortjournal = {Chinese Phys. B},
  volume = {28},
  number = {11},
  pages = {117101},
  publisher = {{Chinese Physical Society and IOP Publishing Ltd}},
  issn = {1674-1056},
  doi = {10.1088/1674-1056/ab4d3b},
  url = {https://dx.doi.org/10.1088/1674-1056/ab4d3b},
  urldate = {2023-06-21},
  langid = {english}
}

@article{avronTopologicalInvariantsFermi1988,
  title = {Topological {{Invariants}} in {{Fermi Systems}} with {{Time-Reversal Invariance}}},
  author = {Avron, J. E. and Sadun, L. and Segert, J. and Simon, B.},
  date = {1988-09-19},
  journaltitle = {Physical Review Letters},
  shortjournal = {Phys. Rev. Lett.},
  volume = {61},
  number = {12},
  pages = {1329--1332},
  issn = {0031-9007},
  doi = {10.1103/PhysRevLett.61.1329},
  url = {https://link.aps.org/doi/10.1103/PhysRevLett.61.1329},
  urldate = {2023-10-17},
  langid = {english}
}

@book{bottDifferentialFormsAlgebraic2008,
  title = {Differential Forms in Algebraic Topology},
  author = {Bott, Raoul and Tu, Loring Wuliang},
  date = {2008},
  series = {Graduate Texts in Mathematics},
  edition = {4. print},
  number = {82},
  publisher = {{Springer}},
  location = {{New York}},
  isbn = {978-0-387-90613-3 978-3-540-90613-1},
  langid = {english},
  pagetotal = {331}
}

@article{burkovTopologicalNodalSemimetals2011,
  title = {Topological Nodal Semimetals},
  author = {Burkov, A. A. and Hook, M. D. and Balents, Leon},
  date = {2011-12-20},
  journaltitle = {Physical Review B},
  shortjournal = {Phys. Rev. B},
  volume = {84},
  number = {23},
  pages = {235126},
  issn = {1098-0121, 1550-235X},
  doi = {10.1103/PhysRevB.84.235126},
  url = {https://link.aps.org/doi/10.1103/PhysRevB.84.235126},
  urldate = {2023-10-17},
  langid = {english}
}

@article{fangTopologicalNodalLine2015a,
  title = {Topological Nodal Line Semimetals with and without Spin-Orbital Coupling},
  author = {Fang, Chen and Chen, Yige and Kee, Hae-Young and Fu, Liang},
  date = {2015-08-10},
  journaltitle = {Physical Review B},
  shortjournal = {Phys. Rev. B},
  volume = {92},
  number = {8},
  pages = {081201},
  issn = {1098-0121, 1550-235X},
  doi = {10.1103/PhysRevB.92.081201},
  url = {https://link.aps.org/doi/10.1103/PhysRevB.92.081201},
  urldate = {2023-10-17},
  langid = {english}
}

@article{freedTwistedEquivariantMatter2013b,
  title = {Twisted {{Equivariant Matter}}},
  author = {Freed, Daniel S. and Moore, Gregory W.},
  date = {2013-12},
  journaltitle = {Annales Henri Poincaré},
  shortjournal = {Ann. Henri Poincaré},
  volume = {14},
  number = {8},
  pages = {1927--2023},
  issn = {1424-0637, 1424-0661},
  doi = {10.1007/s00023-013-0236-x},
  url = {http://link.springer.com/10.1007/s00023-013-0236-x},
  urldate = {2023-10-17},
  langid = {english}
}

@book{hatcherAlgebraicTopology2002,
  title = {Algebraic Topology},
  author = {Hatcher, Allen},
  date = {2002},
  publisher = {{Cambridge University Press}},
  location = {{Cambridge ; New York}},
  isbn = {978-0-521-79160-1 978-0-521-79540-1},
  pagetotal = {544}
}

@article{kimDiracLineNodes2015,
  title = {Dirac {{Line Nodes}} in {{Inversion-Symmetric Crystals}}},
  author = {Kim, Youngkuk and Wieder, Benjamin J. and Kane, C. L. and Rappe, Andrew M.},
  date = {2015-07-17},
  journaltitle = {Physical Review Letters},
  shortjournal = {Phys. Rev. Lett.},
  volume = {115},
  number = {3},
  pages = {036806},
  issn = {0031-9007, 1079-7114},
  doi = {10.1103/PhysRevLett.115.036806},
  url = {https://link.aps.org/doi/10.1103/PhysRevLett.115.036806},
  urldate = {2023-10-19},
  langid = {english}
}

@inproceedings{kitaevPeriodicTableTopological2009a,
  title = {Periodic Table for Topological Insulators and Superconductors},
  booktitle = {{{AIP Conference Proceedings}}},
  author = {Kitaev, Alexei and Lebedev, Vladimir and Feigel’man, Mikhail},
  date = {2009},
  pages = {22--30},
  publisher = {{AIP}},
  location = {{Chernogolokova (Russia)}},
  doi = {10.1063/1.3149495},
  url = {https://pubs.aip.org/aip/acp/article/1134/1/22-30/815164},
  urldate = {2023-10-17},
  eventtitle = {{{ADVANCES IN THEORETICAL PHYSICS}}: {{Landau Memorial Conference}}},
  langid = {english}
}

@book{liuAlgebraicGeometryArithmetic2010a,
  title = {Algebraic Geometry and Arithmetic Curves},
  author = {Liu, Qing},
  date = {2010},
  series = {Oxford Graduate Texts in Mathematics},
  edition = {Reprinted},
  number = {6},
  publisher = {{Oxford Univ. Press}},
  location = {{Oxford}},
  isbn = {978-0-19-920249-2},
  langid = {english},
  pagetotal = {577}
}

@article{mathaiDifferentialTopologySemimetals2017,
  title = {Differential {{Topology}} of {{Semimetals}}},
  author = {Mathai, Varghese and Thiang, Guo Chuan},
  date = {2017-10},
  journaltitle = {Communications in Mathematical Physics},
  shortjournal = {Commun. Math. Phys.},
  volume = {355},
  number = {2},
  pages = {561--602},
  issn = {0010-3616, 1432-0916},
  doi = {10.1007/s00220-017-2965-z},
  url = {http://link.springer.com/10.1007/s00220-017-2965-z},
  urldate = {2023-10-11},
  langid = {english}
}

@article{mathaiGlobalTopologyWeyl2017a,
  title = {Global Topology of {{Weyl}} Semimetals and {{Fermi}} Arcs},
  author = {Mathai, Varghese and Thiang, Guo Chuan},
  date = {2017-03-17},
  journaltitle = {Journal of Physics A: Mathematical and Theoretical},
  shortjournal = {J. Phys. A: Math. Theor.},
  volume = {50},
  number = {11},
  pages = {11LT01},
  issn = {1751-8113, 1751-8121},
  doi = {10.1088/1751-8121/aa59b2},
  url = {https://iopscience.iop.org/article/10.1088/1751-8121/aa59b2},
  urldate = {2023-10-11}
}

@article{matsuuraProtectedBoundaryStates2013,
  title = {Protected Boundary States in Gapless Topological Phases},
  author = {Matsuura, Shunji and Chang, Po-Yao and Schnyder, Andreas P and Ryu, Shinsei},
  date = {2013-06-03},
  journaltitle = {New Journal of Physics},
  shortjournal = {New J. Phys.},
  volume = {15},
  number = {6},
  pages = {065001},
  issn = {1367-2630},
  doi = {10.1088/1367-2630/15/6/065001},
  url = {https://iopscience.iop.org/article/10.1088/1367-2630/15/6/065001},
  urldate = {2023-10-17}
}

@article{satiAnyonicTopologicalOrder2023a,
  title = {Anyonic Topological Order in Twisted Equivariant Differential ({{TED}}) {{K-theory}}},
  author = {Sati, Hisham and Schreiber, Urs},
  date = {2023-04},
  journaltitle = {Reviews in Mathematical Physics},
  shortjournal = {Rev. Math. Phys.},
  volume = {35},
  number = {03},
  pages = {2350001},
  issn = {0129-055X, 1793-6659},
  doi = {10.1142/S0129055X23500010},
  url = {https://www.worldscientific.com/doi/10.1142/S0129055X23500010},
  urldate = {2023-10-18},
  langid = {english}
}

@book{spanierAlgebraicTopology1995,
  title = {Algebraic Topology},
  author = {Spanier, Edwin Henry},
  date = {1995},
  edition = {1. corr. Springer ed},
  publisher = {{Springer}},
  location = {{New York Berlin Barcelona Budapest}},
  isbn = {978-0-387-94426-5 978-3-540-90646-9 978-0-387-90646-1},
  pagetotal = {528}
}

@article{thiangKTheoreticClassificationTopological2016a,
  title = {On the {{K-Theoretic Classification}} of {{Topological Phases}} of {{Matter}}},
  author = {Thiang, Guo Chuan},
  date = {2016-04},
  journaltitle = {Annales Henri Poincaré},
  shortjournal = {Ann. Henri Poincaré},
  volume = {17},
  number = {4},
  pages = {757--794},
  issn = {1424-0637, 1424-0661},
  doi = {10.1007/s00023-015-0418-9},
  url = {http://link.springer.com/10.1007/s00023-015-0418-9},
  urldate = {2023-10-17},
  langid = {english}
}

@book{vanderbiltBerryPhasesElectronic2018,
  title = {Berry {{Phases}} in {{Electronic Structure Theory}}: {{Electric Polarization}}, {{Orbital Magnetization}} and {{Topological Insulators}}},
  shorttitle = {Berry {{Phases}} in {{Electronic Structure Theory}}},
  author = {Vanderbilt, David},
  date = {2018},
  publisher = {{Cambridge University Press}},
  location = {{Cambridge}},
  isbn = {978-1-316-66220-5},
  langid = {english}
}

@article{wanTopologicalSemimetalFermiarc2011a,
  title = {Topological Semimetal and {{Fermi-arc}} Surface States in the Electronic Structure of Pyrochlore Iridates},
  author = {Wan, Xiangang and Turner, Ari M. and Vishwanath, Ashvin and Savrasov, Sergey Y.},
  date = {2011-05-02},
  journaltitle = {Physical Review B},
  shortjournal = {Phys. Rev. B},
  volume = {83},
  number = {20},
  pages = {205101},
  issn = {1098-0121, 1550-235X},
  doi = {10.1103/PhysRevB.83.205101},
  url = {https://link.aps.org/doi/10.1103/PhysRevB.83.205101},
  urldate = {2023-10-17},
  langid = {english}
}

@article{xuDiscoveryWeylFermion2015b,
  title = {Discovery of a {{Weyl}} Fermion Semimetal and Topological {{Fermi}} Arcs},
  author = {Xu, Su-Yang and Belopolski, Ilya and Alidoust, Nasser and Neupane, Madhab and Bian, Guang and Zhang, Chenglong and Sankar, Raman and Chang, Guoqing and Yuan, Zhujun and Lee, Chi-Cheng and Huang, Shin-Ming and Zheng, Hao and Ma, Jie and Sanchez, Daniel S. and Wang, BaoKai and Bansil, Arun and Chou, Fangcheng and Shibayev, Pavel P. and Lin, Hsin and Jia, Shuang and Hasan, M. Zahid},
  date = {2015-08-07},
  journaltitle = {Science},
  shortjournal = {Science},
  volume = {349},
  number = {6248},
  pages = {613--617},
  issn = {0036-8075, 1095-9203},
  doi = {10.1126/science.aaa9297},
  url = {https://www.science.org/doi/10.1126/science.aaa9297},
  urldate = {2023-10-17},
  langid = {english}
}

@article{zhaoSymmetricRealDirac2017,
  title = {P {{T}} -{{Symmetric Real Dirac Fermions}} and {{Semimetals}}},
  author = {Zhao, Y. X. and Lu, Y.},
  date = {2017-02-02},
  journaltitle = {Physical Review Letters},
  shortjournal = {Phys. Rev. Lett.},
  volume = {118},
  number = {5},
  pages = {056401},
  issn = {0031-9007, 1079-7114},
  doi = {10.1103/PhysRevLett.118.056401},
  url = {https://link.aps.org/doi/10.1103/PhysRevLett.118.056401},
  urldate = {2023-10-17},
  langid = {english}
}
 
\end{document}